\documentclass[pdflatex,sn-mathphys-num]{sn-jnl}

\usepackage{graphicx}%
\usepackage{multirow}%
\usepackage{amsmath,amssymb,amsfonts}%
\usepackage{amsthm}%
\usepackage{mathrsfs}%
\usepackage[title]{appendix}%
\usepackage{xcolor}%
\usepackage{textcomp}%
\usepackage{manyfoot}%
\usepackage{booktabs}%
\usepackage{algorithm}%
\usepackage{algcompatible}%
\usepackage{algorithmicx}%
\usepackage{algpseudocode}%
\usepackage{listings}%
\usepackage{braket}%
\usepackage{float}%
\usepackage{dblfloatfix}%
\usepackage{placeins}%
\usepackage{tabularx}%
\usepackage{booktabs}%
\usepackage{array}%

\theoremstyle{thmstyleone}%
\newtheorem{theorem}{Theorem}
\newtheorem{lemma}{Lemma}%

\begin{document}

\title[Article Title]{Capability-Adaptive Cryptanalysis with Reduced-Space Quantum Verification}


\author*[1]{\fnm{Nivedita} \sur{Dey}}\email{mail.dnivedita@gmail.com}

\author[1,2]{\fnm{Mrityunjay} \sur{Ghosh}}\email{mrityunjay.ghosh@hcl.com}
\equalcont{These authors contributed equally to this work.}
\author[2]{\fnm{Pranav} \sur{Kaushal}}\email{pranav\_k@hcl-software.com}

\author[2]{\fnm{Abhinab} \sur{Khare}}\email{abhinavkhare@hcl.com}

\author[1]{\fnm{Amlan} \sur{Chakrabarti}}\email{acakcs@caluniv.ac.in}

\affil*[1]{\orgdiv{AKCSIT}, \orgname{University of Calcutta}, \orgaddress{\city{Kolkata}, \postcode{70106}, \state{West Bengal}, \country{India}}}

\affil[2]{\orgdiv{HCL Software}, \orgname{HCL}, \orgaddress{\city{Noida}, \postcode{201303}, \state{Uttar Pradesh}, \country{India}}}


\abstract{Efficient integration of cryptanalytic evidence with quantum verification remains a fundamental challenge in hybrid classical-quantum cryptanalysis. This work presents a capability-adaptive cryptanalytic framework that unifies linear cryptanalysis, differential cryptanalysis, and side-channel leakage analysis within a common candidate-space reduction architecture, followed by reduced-space quantum verification through amplitude amplification. A formal mathematical model is developed for candidate-space construction, adaptive filtering, verification-space reduction, and complexity characterization, supported by theoretical results establishing the relationship between candidate-space contraction and quantum verification effort. A Hamiltonian formulation is further introduced to provide a physically realizable interpretation of the reduced-space verification process. Evaluation using statistically generated cryptanalytic observations demonstrates that the proposed framework reduces an initial candidate-key hypothesis space of 4096 candidates to an effective candidate space of 13 hypotheses, corresponding to an overall reduction of approximately 99.683\%. Consequently, the Grover verification requirement decreases from 50 iterations to only 2 iterations, yielding an approximately 25-fold reduction in verification effort, while reduced-space amplitude amplification achieves a target-state success probability of approximately 94.53\%. These results demonstrate that adaptive cryptanalytic filtering can substantially reduce quantum verification complexity while preserving cryptanalytic admissibility, providing a practical foundation for capability-aware hybrid cryptanalysis and reduced-space quantum search.}




\maketitle

\section{Introduction}\label{sec1}

The continued reliance on symmetric-key cryptographic systems across digital infrastructures, cloud platforms, communication networks, financial services, and embedded environments underscores the importance of understanding both their security guarantees and their practical attack surfaces. Although modern symmetric ciphers are designed to withstand exhaustive cryptanalysis under well-defined adversarial assumptions, the effectiveness of real-world cryptanalytic efforts is often determined not only by the strength of the underlying algorithm but also by the ability to efficiently identify, prioritize, and verify candidate key hypotheses from large search spaces.\cite{bib1} \cite{bib3} Consequently, contemporary cryptanalysis increasingly extends beyond exhaustive exploration and incorporates statistical inference, leakage exploitation, and data-driven reduction mechanisms that progressively constrain the set of admissible key candidates before final verification is performed. \cite{bib1}

The emergence of quantum computing further broadens this landscape by introducing computational models capable of accelerating specific classes of search and verification problems. \cite{bib7} While symmetric cryptographic systems are generally regarded as more resilient to quantum adversaries than many public-key constructions, the availability of quantum search procedures alters the effective cost associated with candidate-key verification. \cite{bib4} \cite{bib5} At the same time, practical cryptanalytic workflows rarely operate on the complete cryptographic key space. Instead, cryptanalysts frequently rely on statistical distinguishers, differential characteristics, linear approximations, leakage observations, or combinations thereof to eliminate large portions of the search space prior to exhaustive evaluation. This observation suggests that the effectiveness of a cryptanalytic process is influenced not only by the verification mechanism itself but also by the quality and structure of the candidate space presented to that mechanism.

A closely related consideration arises from the manner in which cryptanalytic observations are selected and processed. Differential attacks depend on representative plaintext pairs, linear attacks depend on statistically meaningful observations, and side-channel analyses depend on informative leakage traces. \cite{bib17} In each case, the quality of the underlying observation set directly affects the efficiency of the analytical process. As cryptanalytic datasets continue to grow in scale and complexity, observation-space management becomes increasingly important for reducing computational overhead while preserving analytical effectiveness. Consequently, candidate-space reduction and observation-space reduction may be viewed as complementary objectives that jointly influence the overall efficiency of a cryptanalytic framework.

Recent advances in quantum randomness generation introduce an additional perspective on this problem. Unlike deterministic pseudo-random processes, Quantum Random Number Generators (QRNGs) derive entropy from intrinsically probabilistic quantum phenomena and have therefore attracted significant attention as high-quality randomness sources for security-sensitive applications. \cite{bib14} Beyond their traditional role in cryptographic key generation and secure communication systems, quantum-generated randomness offers potential advantages for unbiased observation selection and statistically robust sampling procedures. Although such mechanisms do not alter the cryptographic strength of a symmetric cipher, they may contribute to the construction of representative observation subsets and reduce selection bias within cryptanalytic evaluation processes. This observation motivates the exploration of quantum-assisted sampling as a supporting component within broader cryptanalytic workflows.

Motivated by these considerations, this work investigates a unified cryptanalytic framework that combines adaptive candidate-space reduction, capability-driven attack selection, reduced-space quantum verification, and quantum-assisted observation sampling within a single analytical workflow. \cite{bib15} Rather than treating statistical cryptanalysis, leakage-guided reduction, quantum verification, and observation selection as isolated stages, the proposed framework considers their interaction and collective influence on candidate-key identification. \cite{bib10} The resulting formulation provides a structured mechanism for integrating multiple cryptanalytic perspectives while maintaining a consistent mathematical representation of candidate-space evolution and verification complexity.\cite{bib2}

The remainder of this paper is organized as follows. Section~2 reviews existing research related to quantum cryptanalysis, statistical cryptanalytic techniques, observation-space reduction, and quantum-generated randomness. \cite{bib11} \cite{bib15} Section~3 develops a unified mathematical framework for candidate-space construction and reduction using linear, differential, and leakage-guided cryptanalytic models. Section~4 introduces the adaptive hybrid cryptanalytic architecture, presents its theoretical formulation, derives the associated complexity bounds, and develops a reduced-space quantum verification procedure. Section~5 examines the role of quantum-assisted observation sampling and discusses its integration within cryptanalytic workflows. Finally, Section~6 summarizes the principal findings and outlines directions for future investigation.

\section{Literature Review}\label{sec2}

The security of modern digital infrastructures continues to rely heavily on symmetric-key cryptographic mechanisms owing to their computational efficiency and suitability for large-scale data protection. Among these mechanisms, the Advanced Encryption Standard (AES) has emerged as the dominant encryption standard across communication networks, cloud infrastructures, financial systems, and resource-constrained environments. Comparative studies involving AES, DES, and 3DES consistently demonstrate the favorable balance between security, performance, and implementation efficiency achieved by AES, leading to its widespread adoption in both commercial and security-critical applications \cite{bib1}. Contemporary cryptographic research, however, increasingly recognizes that practical security is influenced not only by algorithmic design but also by implementation characteristics, attack models, and available computational capabilities.

The rapid advancement of quantum computing has introduced a new threat model for cryptographic systems. While Shor's algorithm primarily affects public-key cryptography, Grover's algorithm provides a quadratic speedup for exhaustive key-search problems and therefore directly impacts the effective security margin of symmetric encryption schemes. This observation has motivated extensive investigations into quantum cryptanalysis of block ciphers, including quantum oracle constructions, quantum key-search procedures, reversible AES implementations, and reduced-complexity quantum search strategies \cite{bib2} \cite{bib6} \cite{bib10}. Existing studies indicate that symmetric cryptographic algorithms remain comparatively resilient in the near term, although their effective security level may be reduced under sufficiently powerful quantum adversaries. \cite{bib9} Consequently, considerable attention has been directed toward cryptanalytic techniques capable of reducing the practical effort required for candidate-key identification and verification.

Beyond exhaustive search, statistical cryptanalysis remains one of the principal approaches for evaluating cipher security. \cite{bib13} Differential cryptanalysis, linear cryptanalysis, and side-channel analysis exploit statistical distinguishers, probability distributions, leakage information, and observable execution characteristics to eliminate inconsistent key hypotheses prior to final verification. \cite{bib1} \cite{bib3} These methodologies demonstrate that cryptanalysis can be interpreted as a progressive candidate-space reduction process, where successive analytical stages iteratively contract the admissible key space. Such reduction-oriented perspectives have become increasingly important in modern attack frameworks because they enable computational resources to be concentrated on statistically promising key hypotheses.

The efficiency of practical cryptanalysis is also strongly influenced by the construction and processing of the observation space. Several studies have shown that complete evaluation of all available observations is often computationally unnecessary, particularly in large-scale statistical analyses. Sampling-based cryptanalysis, feature-guided evaluation, machine-learning-assisted differential analysis, and observation-selection strategies demonstrate that carefully selected subsets of observations can preserve analytical effectiveness while substantially reducing computational overhead \cite{bib21}. These findings suggest that observation reduction and candidate-space contraction constitute complementary mechanisms for improving the scalability of cryptanalytic workflows, especially when large volumes of statistical data must be processed.

In parallel, significant research efforts have focused on the generation and utilization of high-quality randomness within security-critical systems. Quantum Random Number Generators (QRNGs) exploit fundamentally probabilistic quantum phenomena to generate entropy that is not derived from deterministic computational processes. \cite{bib18} Consequently, QRNGs have been investigated for secure communication infrastructures, authentication frameworks, post-quantum security architectures, key-generation mechanisms, and entropy-enhanced cryptographic applications. \cite{bib12} The resulting quantum-generated randomness exhibits desirable statistical properties that can support applications requiring unbiased selection, uncertainty-driven decision making, and statistically robust sampling procedures.

The statistical characteristics of quantum-generated randomness motivate its consideration beyond conventional key-generation scenarios. Observation selection, trace selection, pair selection, and statistical estimation frequently influence the effectiveness of cryptanalytic procedures, particularly in differential, linear, and leakage-based analyses. Recent studies on entropy characterization, randomness evaluation, and quantum-random sequence analysis suggest that quantum-generated observations may provide advantageous properties for unbiased sampling and statistically representative observation construction. \cite{bib8} Although the integration of QRNGs into cryptographic infrastructures has been extensively explored, their potential role in adaptive cryptanalytic observation selection and statistical evaluation remains comparatively underexamined.

Collectively, the existing literature demonstrates substantial progress in four closely related research directions: quantum cryptanalysis of symmetric cryptographic systems, statistical and leakage-based cryptanalytic attacks, observation-space reduction techniques, and quantum-generated randomness. \cite{bib22} Nevertheless, these directions have largely evolved independently. \cite{bib1} \cite{bib2} Existing studies typically investigate cryptanalytic reduction, quantum search acceleration, sampling-based optimization, or QRNG-enabled security enhancement in isolation. \cite{bib11} Limited attention has been devoted to a unified framework that systematically combines adaptive cryptanalytic filtering, capability-driven attack selection, reduced candidate-key construction, reduced-space quantum verification, and quantum-assisted observation sampling within a single analytical workflow. Addressing this integration gap motivates the framework developed in the subsequent sections, where classical cryptanalytic reduction, adaptive candidate-space contraction, quantum verification, and QRNG-assisted observation sampling are integrated into a coherent hybrid cryptanalytic architecture.

\section{Unified Classical-Quantum Cryptanalytic Model and Baseline Candidate-Space Construction}

The proposed framework models a symmetric-key cryptographic primitive as

\begin{equation}
\Pi=
(\mathcal P_n,\mathcal C_n,\mathcal K_k,\mathcal E_K,\mathcal D_K),
\label{eq:crypto_model}
\end{equation}

where $\mathcal P_n=\{0,1\}^{n}$, $\mathcal C_n=\{0,1\}^{n}$, and $\mathcal K_k=\{0,1\}^{k}$ denote the plaintext space, ciphertext space, and secret-key space, respectively. The encryption and decryption mappings satisfy the correctness condition

\begin{equation}
\mathcal D_K\!\left(\mathcal E_K(P)\right)=P,
\qquad
\forall P\in\mathcal P_n,
\quad
\forall K\in\mathcal K_k.
\label{eq:correctness}
\end{equation}

The formulation in (\ref{eq:crypto_model}) and (\ref{eq:correctness}) is intentionally cipher-agnostic and therefore remains applicable to substitution-permutation networks, Feistel constructions, and hybrid symmetric-key architectures. To establish a unified analytical environment reused throughout all subsequent cryptanalytic procedures, plaintext and ciphertext differences are represented by $\Delta P=P_1\oplus P_2$ and $\Delta C=C_1\oplus C_2$, respectively. Differential distinguishability is characterized through the differential propagation operator

\begin{equation}
DP(\Delta P,\Delta C)
=
\Pr
\Big[
\mathcal E_K(P_1)
\oplus
\mathcal E_K(P_1\oplus\Delta P)
=
\Delta C
\Big].
\label{eq:differential_probability}
\end{equation}

Similarly, linear distinguishers are represented through the Boolean approximation mapping

\begin{equation}
\Gamma:
\mathcal P_n
\times
\mathcal C_n
\times
\mathcal K_k
\rightarrow
\{0,1\},
\label{eq:linear_mapping}
\end{equation}

with statistical bias

\begin{equation}
\epsilon
=
\left|
\Pr[\Gamma(P,C,K)=0]
-\frac12
\right|.
\label{eq:linear_bias}
\end{equation}

Implementation-level leakage is modeled through the generalized observation process

\begin{equation}
L(P,K,t)
=
f(P,K)+\eta(t),
\label{eq:leakage_model}
\end{equation}

where $f(P,K)$ denotes deterministic key-dependent leakage and $\eta(t)$ denotes stochastic environmental noise. Consequently, the framework simultaneously captures timing leakage, power-analysis leakage, electromagnetic emanations, cache-observation channels, and microarchitectural side-channel effects within a common abstraction. \cite{bib16} For quantum-enabled adversarial analysis, the computational environment is extended into the Hilbert space $H=\mathbb C^{2^m},\label{eq:hilbert_space}$ where $m$ denotes the total qubit count. Candidate-key states, plaintext states, ciphertext states, and auxiliary workspaces are represented by the quantum registers $\ket K$, $\ket P$, $\ket C$, and $\ket A$, respectively. The corresponding verification oracle is defined through $U_f\ket K = (-1)^{f(K)}\ket K,
\label{eq:oracle_definition}$ where $f(K) = \begin{cases} 1, & \mathcal E_K(P)=C, \\0, &
\text{otherwise}.\end{cases}\label{eq:oracle_function}$

The resulting mathematical environment is reused throughout candidate-space construction, linear cryptanalysis, differential cryptanalysis, leakage-guided filtering, and reduced-space quantum amplification without requiring subsequent symbol redefinition. \cite{bib4} \cite{bib6} Rather than directly targeting secret-key recovery, the proposed framework first constructs a consistency-preserving candidate-key subset. Given an observed plaintext-ciphertext pair $(P,C)$, the objective is to identify $\mathcal K_{\mathrm{cand}}
= \Big\{K_i\in\mathcal K_k:\mathcal E_{K_i}(P)=C \Big\},\label{eq:candidate_set}$ which represents the admissible cryptanalytic search space prior to the application of statistical distinguishers and leakage-guided reduction mechanisms. Unlike conventional brute-force formulations that terminate immediately after locating a valid key hypothesis, the proposed formulation preserves the complete candidate subset because all subsequent cryptanalytic modules operate directly on $\mathcal K_{\mathrm{cand}}$.

The procedure CandidateSpaceConstruction$(P,C,\mathcal K_k,\mathcal D)$ constructs the initial admissible key-hypothesis space by exhaustively evaluating cryptographic consistency over the complete key domain. Each candidate key is independently verified against the observed plaintext-ciphertext pair and all consistency-preserving hypotheses are accumulated into the candidate subset $\mathcal K_{\mathrm{cand}}$. The resulting subset serves as the parent search space subsequently consumed by linear, differential, leakage-guided, and quantum-amplified reduction procedures.

\begin{algorithm}[H]
\caption{CandidateSpaceConstruction$(P,C,\mathcal K_k,\mathcal D)$}
\label{alg:candidate_space_construction}

\begin{algorithmic}[1]

\State Initialize
\State $\mathcal K_{\mathrm{cand}}\gets\emptyset$

\ForAll{$K_i\in\mathcal K_k$}

    \State $P_i' \gets \mathcal D_{K_i}(C)$

    \If{$P_i'=P$}
        \State $\mathcal K_{\mathrm{cand}}
        \gets
        \mathcal K_{\mathrm{cand}}
        \cup
        \{K_i\}$
    \EndIf

\EndFor

\State \Return $\mathcal K_{\mathrm{cand}}$

\end{algorithmic}
\end{algorithm}

\paragraph{Computational complexity analysis}

Let

\begin{equation}
N_K
=
|\mathcal K_k|
=
2^k,
\label{eq:keyspace_size}
\end{equation}

denote the cardinality of the candidate-key space. The initialization phase of Algorithm~\ref{alg:candidate_space_construction} requires only candidate-set allocation and therefore contributes

\begin{equation}
T_1=\theta(1)
\label{eq:init_complexity}
\end{equation}

The dominant computational cost arises from candidate-key enumeration. Since exactly $N_K$ candidate keys are evaluated, the decryption phase performs $P_i'
=\mathcal D_{K_i}(C),\label{eq:decryption_step}$ for every candidate key. Let $T_D(n)$ denote the computational complexity of a single decryption evaluation. The total decryption cost therefore becomes

\begin{equation}
T_3
=
N_K\,T_D(n).
\label{eq:decryption_complexity}
\end{equation}

Similarly, the consistency verification stage performs one plaintext-comparison operation per candidate key. Let $T_V(n)$ denote the corresponding verification complexity. Since comparison requires at most $n$ bit evaluations, $T_V(n)=\theta(n),
\label{eq:verification_complexity}$

yielding

\begin{equation}
T_4
=
N_K\,T_V(n).
\label{eq:total_verification_complexity}
\end{equation}

Combining (\ref{eq:init_complexity}), (\ref{eq:decryption_complexity}), and (\ref{eq:total_verification_complexity}), the total execution time becomes $T_{\mathrm{CSC}} = T_1+T_3+T_4
= \theta(1)+N_K\Big(T_D(n)+T_V(n)\Big). \label{eq:total_complexity}$ For practical symmetric-key cryptographic systems, cryptographic evaluation dominates comparison cost, i.e., $T_D(n)\gg T_V(n),
\label{eq:dominance}$ and therefore $T_{\mathrm{CSC}}
= \theta\Big(N_K\,T_D(n)\Big). \label{eq:simplified_complexity}$ Substituting (\ref{eq:keyspace_size}) into (\ref{eq:simplified_complexity}) yields $T_{\mathrm{CSC}}=\theta\Big(2^k\,T_D(n)\Big).
\label{eq:keyspace_complexity}$ Assuming fixed-round block-cipher evaluation scales linearly with block length, $T_D(n) = \theta(n), \label{eq:block_complexity}$ the resulting complexity becomes $T_{\mathrm{CSC}} = \theta(n\,2^k).
\label{eq:linearized_complexity}$ Since $n\ll 2^k$ for practical cryptographic systems, the asymptotic behavior is dominated by the exponential search term, yielding $T_{\mathrm{CSC}} = \theta(2^k).
\label{eq:final_complexity}$ The corresponding worst-case complexity is $T_{\mathrm{CSC}}=O(2^k),
\label{eq:worst_case}$ while the lower bound becomes $T_{\mathrm{CSC}}=\Omega(1),
\label{eq:best_case}$ under idealized early-success conditions. The memory complexity is determined by the cardinality of the candidate subset,n $M_{\mathrm{CSC}} = O(|\mathcal K_{\mathrm{cand}}|),
\label{eq:memory_complexity}$ which reduces to $O(1)$ for a unique valid key and reaches $O(2^k)$ in the theoretical worst case. Consequently, CandidateSpaceConstruction$(\cdot)$ establishes the baseline computational security bound of the target symmetric-key primitive and serves as the reference complexity against which the subsequent candidate-space reduction procedures and reduced-space quantum amplification framework are evaluated.

\subsection{Statistical candidate-space reduction framework}

Following construction of the admissible candidate-key subset $\mathcal K_{\mathrm{cand}}$ using Algorithm~\ref{alg:candidate_space_construction}, the framework performs progressive candidate-space reduction through statistical distinguishers and implementation-level leakage analysis. Unlike candidate-space construction, which preserves all cryptographically consistent key hypotheses, the present stage attempts to eliminate statistically inconsistent candidates by exploiting measurable deviations from ideal random behaviour. Three complementary reduction mechanisms are considered, namely linear cryptanalysis, differential cryptanalysis, and leakage-guided cryptanalysis. \cite{bib20} \cite{bib21} Each mechanism operates directly on $\mathcal K_{\mathrm{cand}}$ and produces a reduced candidate subset that is subsequently consumed by the adaptive hybrid framework.

\paragraph{Linear candidate-space reduction}

Linear cryptanalysis exploits probabilistic linear relations between plaintext variables, ciphertext variables, and key-dependent intermediate states. Let $\Gamma(P,C,K)$ denote the Boolean approximation function introduced in (\ref{eq:linear_mapping}) and let $\epsilon$ denote the corresponding statistical bias defined in (\ref{eq:linear_bias}). For an ideal random permutation, $\epsilon=0$, whereas non-zero bias values indicate exploitable statistical distinguishers. Given an observed plaintext-ciphertext dataset $\mathcal S = \left\{(P_1,C_1),
(P_2,C_2),\dots,(P_N,C_N)\right\}, \label{eq:linear_dataset}$ each candidate key $K_i\in\mathcal K_{\mathrm{cand}}$ is evaluated against the selected approximation mask. Let $N_i$ denote the number of approximation satisfactions associated with candidate key $K_i$. The empirical approximation bias is therefore estimated as $\hat{\epsilon}_i=\left|\frac{N_i}{N}-\frac12\right|$. Candidate keys satisfying the significance criterion $\hat{\epsilon}_i\ge
\epsilon_{\min}\label{eq:linear_threshold}$ are retained, yielding the reduced candidate subset $\mathcal K_{\mathrm{linear}} = \left\{K_i\in
\mathcal K_{\mathrm{cand}}:\hat{\epsilon}_i\ge \epsilon_{\min}\right\}.\label{eq:klinear}$ The procedure LinearCandidateReduction$(\mathcal K_{\mathrm{cand}},\mathcal S,\Gamma,\epsilon_{\min})$ evaluates the statistical consistency of every candidate key with respect to the selected linear approximation and retains only those hypotheses exhibiting measurable approximation bias above the predefined significance threshold.

\begin{algorithm}[H]
\caption{LinearCandidateReduction$(\mathcal K_{\mathrm{cand}},\mathcal S,\Gamma,\epsilon_{\min})$}
\label{alg:linear_candidate_reduction}

\begin{algorithmic}[1]

\State Initialize
\State $\mathcal K_{\mathrm{linear}}\gets\emptyset$

\ForAll{$K_i\in\mathcal K_{\mathrm{cand}}$}

      \State Evaluate $\Gamma(P_j,C_j,K_i)$ for all $(P_j,C_j)\in\mathcal S$

      \State Compute empirical bias $\hat{\epsilon}_i$

      \If{$\hat{\epsilon}_i\ge\epsilon_{\min}$}

            \State $\mathcal K_{\mathrm{linear}}
            \gets
            \mathcal K_{\mathrm{linear}}
            \cup
            \{K_i\}$

      \EndIf

\EndFor

\State \Return $\mathcal K_{\mathrm{linear}}$

\end{algorithmic}
\end{algorithm}

The resulting procedure transforms linear cryptanalysis from a key-ranking mechanism into a candidate-space filtering mechanism suitable for adaptive hybrid cryptanalysis.

\paragraph{Differential candidate-space reduction}

Differential cryptanalysis evaluates the propagation behaviour of input differences through the cryptographic transformation. Using the notation introduced in (\ref{eq:differential_probability}), let $\Delta P=P_1\oplus P_2$ and $\Delta C=C_1\oplus C_2$ denote the selected plaintext and ciphertext differences, respectively. The corresponding differential propagation probability is represented by $DP(\Delta P,\Delta C) = \Pr\Big[\mathcal E_K(P)
\oplus\mathcal E_K(P\oplus\Delta P) = \Delta C\Big].
\label{eq:dp_reduction}$ Given a chosen-plaintext dataset $\mathcal P_{\Delta} = \left\{ (P_j,P_j\oplus\Delta P)\right\}_{j=1}^{N},
\label{eq:differential_dataset}$ each candidate key is evaluated according to the observed agreement between predicted and measured ciphertext differentials. Let $\widehat{DP}_i$ denote the empirical propagation probability associated with candidate key $K_i$. Candidate keys satisfying $\widehat{DP}_i\ge DP_{\min}
\label{eq:differential_threshold}$ are retained, producing the reduced subset $\mathcal K_{\mathrm{differential}} = \left\{K_i\in\mathcal K_{\mathrm{cand}}:\widehat{DP}_i\ge DP_{\min}\right\}.
\label{eq:kdifferential}$ The procedure DifferentialCandidateReduction$(\mathcal K_{\mathrm{cand}},\mathcal P_{\Delta},\Delta P,\Delta C,DP_{\min})$ evaluates candidate-key consistency with respect to the selected differential characteristic and eliminates hypotheses exhibiting propagation behaviour inconsistent with the expected differential trail.

\begin{algorithm}[H]
\caption{DifferentialCandidateReduction$(\mathcal K_{\mathrm{cand}},\mathcal P_{\Delta},\Delta P,\Delta C,DP_{\min})$}
\label{alg:differential_candidate_reduction}

\begin{algorithmic}[1]

\State Initialize
\State $\mathcal K_{\mathrm{differential}}\gets\emptyset$

\ForAll{$K_i\in K_{\mathrm{cand}}$}

      \State Evaluate differential propagation over all plaintext pairs in $\mathcal P_{\Delta}$

      \State Compute empirical propagation probability $\widehat{DP}_i$

      \If{$\widehat{DP}_i\ge DP_{\min}$}

            \State $\mathcal K_{\mathrm{differential}}
            \gets
            \mathcal K_{\mathrm{differential}}
            \cup
            \{K_i\}$

      \EndIf

\EndFor

\State \Return $\mathcal K_{\mathrm{differential}}$

\end{algorithmic}
\end{algorithm}

The resulting candidate subset $\mathcal K_{\mathrm{differential}}$ contains only those key hypotheses exhibiting differential behaviour compatible with the selected characteristic and therefore provides an additional reduction of the admissible search space prior to leakage-guided analysis.

\paragraph{Leakage-guided candidate-space reduction}

While linear and differential cryptanalysis exploit statistical properties of the cryptographic transformation, leakage-guided cryptanalysis exploits implementation-dependent physical observables.\cite{bib20} Using the generalized leakage model introduced in (\ref{eq:leakage_model}), let $\mathcal L = \left\{
(P_1,\tau_1),(P_2,\tau_2),\dots,(P_N,\tau_N)\right\},
\label{eq:leakage_dataset}$ denote the observed leakage dataset, where $\tau_i$ represents the measured leakage associated with plaintext observation $P_i$. For each candidate key $K_i\in\mathcal K_{\mathrm{cand}}$, a hypothetical leakage vector $H(P,K_i)$ is generated according to the selected leakage model. The statistical agreement between the hypothetical leakage vector and the observed leakage measurements is quantified using the Pearson correlation coefficient

\begin{equation}
\rho_i
=
\rho
\Big(
H(P,K_i),
\mathcal L
\Big).
\label{eq:pearson_correlation}
\end{equation}

Candidate keys satisfying $\rho_i\ge\rho_{\min}
\label{eq:leakage_threshold}$ are retained, yielding
$\mathcal K_{\mathrm{leakage}} = \left\{K_i\in
\mathcal K_{\mathrm{cand}}:\rho_i\ge\rho_{\min}
\right\}.\label{eq:kleakage}$ The procedure LeakageGuidedCandidateReduction$(\mathcal K_{\mathrm{cand}},\mathcal L,H,\rho_{\min})$ evaluates the statistical consistency between observed leakage measurements and hypothetical leakage estimates generated for each candidate key. Only candidate keys exhibiting sufficiently strong correlation with the measured leakage traces are preserved.

\begin{algorithm}[H]
\caption{LeakageGuidedCandidateReduction$(\mathcal K_{\mathrm{cand}},\mathcal L,H,\rho_{\min})$}
\label{alg:leakage_candidate_reduction}

\begin{algorithmic}[1]

\State Initialize
\State $\mathcal K_{\mathrm{leakage}}\gets\emptyset$

\ForAll{$K_i\in\mathcal K_{\mathrm{cand}}$}

      \State Generate hypothetical leakage vector
      $H_i \gets H(P,K_i)$

      \State Compute correlation coefficient $\rho_i$

      \If{$\rho_i\ge\rho_{\min}$}

            \State $\mathcal K_{\mathrm{leakage}}
            \gets
            \mathcal K_{\mathrm{leakage}}
            \cup
            \{K_i\}$

      \EndIf

\EndFor

\State \Return $\mathcal K_{\mathrm{leakage}}$

\end{algorithmic}
\end{algorithm}

The resulting candidate subset $\mathcal K_{\mathrm{leakage}}$ contains only those key hypotheses exhibiting statistically significant agreement with the observed physical leakage behaviour and therefore provides an implementation-aware reduction of the admissible cryptanalytic search space.

\subsection{Complexity analysis}

Let $N_K = |\mathcal K_{\mathrm{cand}}| \label{eq:candidate_cardinality}$ denote the cardinality of the candidate-key subset generated by Algorithm~\ref{alg:candidate_space_construction}, and let $N$ denote the number of observations available for statistical analysis. For Algorithm~\ref{alg:linear_candidate_reduction}, the dominant computational cost arises from repeated evaluation of the approximation function $\Gamma(P_j,C_j,K_i)$. Let $T_{\Gamma}$ denote the cost of evaluating a single approximation mask. Since the approximation consists of a finite number of XOR operations over selected plaintext, ciphertext, and key bits, the evaluation cost remains constant, i.e., $T_{\Gamma}=\theta(1)$. Consequently, the total execution time becomes $T_{\mathrm{linear}} = N_KN
T_{\Gamma}.\label{eq:linear_complexity_1}$ Substituting $T_{\Gamma}=\theta(1)$ yields $T_{\mathrm{linear}} = \theta(N_KN). \label{eq:linear_complexity}$ The corresponding memory complexity is $M_{\mathrm{linear}} = O(N_K).
\label{eq:linear_memory}$ For Algorithm~\ref{alg:differential_candidate_reduction}, each candidate key is evaluated over all plaintext pairs contained in $\mathcal P_{\Delta}$. Let $T_E(n)$ denote the complexity of a single encryption evaluation. Since each differential test requires computation of \[C_1=\mathcal E_{K_i}(P),\qquad
C_2=\mathcal E_{K_i}(P\oplus\Delta P),\] two encryption evaluations are performed per observation. Therefore, $T_{\mathrm{differential}} = 2N_KN\,T_E(n).
\label{eq:differential_complexity_1}$ For practical fixed-round symmetric ciphers, encryption complexity scales linearly with block size, yielding $T_E(n)=\theta(n)$. Consequently, $T_{\mathrm{differential}} = \theta(N_KNn).
\label{eq:differential_complexity}$ The corresponding memory complexity becomes $M_{\mathrm{differential}} =
O(N_K).\label{eq:differential_memory}$ For Algorithm~\ref{alg:leakage_candidate_reduction}, two dominant operations are performed for each candidate key. First, a hypothetical leakage vector is generated over all observations. Let $T_H(N)$ denote the corresponding generation cost. Since one hypothetical leakage value is produced per observation,

\begin{equation}
T_H(N)
=
\theta(N).
\label{eq:hypothetical_cost}
\end{equation}

Second, the Pearson correlation coefficient in (\ref{eq:pearson_correlation}) is evaluated. Let $T_{\rho}(N)$ denote the corresponding computational cost. Since covariance and variance estimation require a single traversal of the observation vectors,

\begin{equation}
T_{\rho}(N)
=
\theta(N).
\label{eq:correlation_cost}
\end{equation}

Therefore, the total execution time becomes $T_{\mathrm{leakage}} = N_K\Big(T_H(N)+T_{\rho}(N)
\Big).\label{eq:leakage_complexity_1}$ Substituting (\ref{eq:hypothetical_cost}) and (\ref{eq:correlation_cost}) yields $T_{\mathrm{leakage}} = \theta(N_KN). \label{eq:leakage_complexity}$ The corresponding memory complexity becomes $M_{\mathrm{leakage}} = O(N).\label{eq:leakage_memory}$ Collectively, the three reduction procedures transform the original candidate subset into the effective search space $\mathcal K_{\mathrm{eff}} = \mathcal K_{\mathrm{cand}}\cap\mathcal K_{\mathrm{linear}}\cap
\mathcal K_{\mathrm{differential}}\cap\mathcal K_{\mathrm{leakage}}.\label{eq:effective_candidate_space}$ The corresponding reduced search-space cardinality becomes $N_{\mathrm{eff}} = |\mathcal K_{\mathrm{eff}}|.\label{eq:effective_cardinality}$ Consequently, the objective of the present stage is not direct key recovery but progressive candidate-space reduction. The reduced search-space cardinality $N_{\mathrm{eff}}$ subsequently becomes the primary control parameter governing the complexity of the adaptive hybrid cryptanalytic framework and the reduced-space quantum amplification procedure presented in the following section.

\section{Adaptive Quantum-Enhanced Hybrid Cryptanalytic Framework}
\label{sec:adaptive_framework}

The candidate-space construction and cryptanalytic reduction procedures developed in Chapter~3 establish statistically admissible key hypotheses obtained through linear, differential, and leakage-guided cryptanalysis. However, the applicability of individual distinguishers depends upon the practical adversarial environment. In realistic cryptanalytic settings, known-plaintext availability, chosen-plaintext access, leakage observability, and quantum computational capability are not universally available. Consequently, a static attack strategy may invoke inapplicable distinguishers or fail to exploit available information efficiently. To address this limitation, an adaptive hybrid cryptanalytic framework is introduced in which cryptanalytic modules are dynamically activated according to adversarial capabilities and statistical distinguishability conditions. Let the adversarial capability set be represented by $\mathcal{A} = \left\{\mathcal{A}_{\mathrm{known}},\mathcal{A}_{\mathrm{chosen}},\mathcal{A}_{\mathrm{leakage}},\mathcal{A}_{\mathrm{quantum}},\mathcal{A}_{\mathrm{oracle}}\right\},
\label{eq:capability_set}$ where the individual elements respectively denote known-plaintext access, chosen-plaintext capability, leakage observability, quantum computational availability, and reversible oracle-construction feasibility. Rather than activating every cryptanalytic procedure unconditionally, the framework associates each reduction mechanism with a capability-dependent activation predicate. Accordingly, the linear, differential, and leakage-guided activation indicators are defined as

\begin{equation}
\chi_{\mathrm{L}}
=
\mathbf{1}
\!\left[
\mathcal{A}_{\mathrm{known}}
\land
\left(
\epsilon_{\mathrm{linear}}
>
\epsilon_{\min}
\right)
\right],
\label{eq:linear_activation}
\end{equation}

\begin{equation}
\chi_{\mathrm{D}}
=
\mathbf{1}
\!\left[
\mathcal{A}_{\mathrm{chosen}}
\land
\left(
DP
>
DP_{\min}
\right)
\right],
\label{eq:differential_activation}
\end{equation}

\begin{equation}
\chi_{\mathrm{S}}
=
\mathbf{1}
\!\left[
\mathcal{A}_{\mathrm{leakage}}
\land
\left(
\rho_{\mathrm{leakage}}
>
\rho_{\min}
\right)
\right],
\label{eq:leakage_activation}
\end{equation}

where $\mathbf{1}[\cdot]$ denotes the indicator function. The activation indicators determine whether the corresponding reduction procedures defined in Chapter~3 participate in candidate-space refinement. Let $\mathcal{K}_{\mathrm{cand}}$ denote the candidate-key space generated by \textsc{CandidateSpaceConstruction}$(\cdot)$ and let the activated reduction procedures produce candidate subsets $\mathcal{K}_{\mathrm{linear}}$, $\mathcal{K}_{\mathrm{differential}}$, and $\mathcal{K}_{\mathrm{leakage}}$. The adaptive candidate-space refinement process is represented by $\begin{aligned}\mathcal{K}_{\mathrm{eff}} = &
\mathcal{K}_{\mathrm{cand}}\\&\cap\left(\chi_{\mathrm{L}}\mathcal{K}_{\mathrm{linear}}+(1-\chi_{\mathrm{L}})
\mathcal{K}_{\mathrm{cand}}\right)\\&\cap\left(\chi_{\mathrm{D}}\mathcal{K}_{\mathrm{differential}}+(1-\chi_{\mathrm{D}})\mathcal{K}_{\mathrm{cand}}\right)\\
&\cap\left(\chi_{\mathrm{S}}\mathcal{K}_{\mathrm{leakage}}+(1-\chi_{\mathrm{S}})\mathcal{K}_{\mathrm{cand}}
\right),\end{aligned}\label{eq:adaptive_candidate_space}$ which guarantees that only activated distinguishers contribute to candidate-space reduction while inactive modules leave the admissible key space unchanged. The resulting effective search-space cardinality becomes $N_{\mathrm{eff}} = \left|
\mathcal{K}_{\mathrm{eff}}\right|,\label{eq:effective_search_space}$

which represents the reduced cryptanalytic search domain produced through adaptive statistical filtering. If $N_{\mathrm{eff}}=0$, no candidate key simultaneously satisfies the activated distinguishers and the attack terminates unsuccessfully. Otherwise, $\mathcal{K}_{\mathrm{eff}}$ becomes the input to the quantum verification stage. Quantum verification is activated only when both quantum computational resources and reversible oracle-construction capability are available. Accordingly, the quantum activation predicate is defined as

\begin{equation}
\chi_{\mathrm{Q}}
=
\mathbf{1}
\!\left[
\mathcal{A}_{\mathrm{quantum}}
\land
\mathcal{A}_{\mathrm{oracle}}
\right],
\label{eq:quantum_activation}
\end{equation}

where $\chi_{\mathrm{Q}}=1$ indicates that reduced-space quantum verification is feasible.

The effective quantum search dimension is subsequently determined as

\begin{equation}
q
=
\left\lceil
\log_{2}
N_{\mathrm{eff}}
\right\rceil,
\label{eq:effective_qubit_dimension}
\end{equation}

which specifies the minimum number of computational qubits required to encode the reduced candidate-key space. Unlike exhaustive quantum key search operating over the complete key space $\mathcal{K}_{k}$, the proposed framework performs quantum verification exclusively over $\mathcal{K}_{\mathrm{eff}}$, thereby coupling classical cryptanalytic reduction with reduced-space quantum amplitude amplification within a unified capability-adaptive attack model.

Algorithm~\ref{alg:adaptive_hybrid_key_recovery} presents the complete adaptive orchestration procedure responsible for coordinating candidate-space construction, capability-driven cryptanalytic activation, statistical key-space reduction, and reduced-space quantum verification.
\begin{algorithm}[t]
\caption{AdaptiveHybridKeyRecovery$(\mathcal{K}_{k},\mathcal{S},\mathcal{P}_{\Delta},\mathcal{L},\epsilon_{\min},DP_{\min},\rho_{\min},\mathcal{A})$}
\label{alg:adaptive_hybrid_key_recovery}

\REQUIRE{
Complete key space $\mathcal{K}_{k}$,
known plaintext-ciphertext observations $\mathcal{S}$,
chosen plaintext pair set $\mathcal{P}_{\Delta}$,
leakage observation set $\mathcal{L}$,
activation thresholds $\epsilon_{\min}$, $DP_{\min}$, $\rho_{\min}$,
adversarial capability set $\mathcal{A}$
}

\ENSURE{
Recovered key hypothesis $K^{*}$
}

$\mathcal{K}_{\mathrm{cand}}
\leftarrow
\textsc{CandidateSpaceConstruction}
(\mathcal{K}_{k})$\;

Compute
$\chi_{\mathrm{L}}$,
$\chi_{\mathrm{D}}$,
$\chi_{\mathrm{S}}$,
$\chi_{\mathrm{Q}}$
using
Eqs.~(\ref{eq:linear_activation})-
(\ref{eq:quantum_activation})\;
\begin{algorithmic}
\If{$\chi_{\mathrm{L}}=1$}
\state{
$\mathcal{K}_{\mathrm{linear}}
\leftarrow
\textsc{LinearCandidateReduction}
(\mathcal{K}_{\mathrm{cand}},\mathcal{S},\epsilon_{\min})$\;
}
\Else
\state{
$\mathcal{K}_{\mathrm{linear}}
\leftarrow
\mathcal{K}_{\mathrm{cand}}$\;
}
\EndIf
\end{algorithmic}
\begin{algorithmic}
\If{$\chi_{\mathrm{D}}=1$}
\state{$\mathcal{K}_{\mathrm{differential}}
\leftarrow
\textsc{DifferentialCandidateReduction}
(\mathcal{K}_{\mathrm{cand}},\mathcal{P}_{\Delta},DP_{\min})$\;
}
\Else
\state{
$\mathcal{K}_{\mathrm{differential}}
\leftarrow
\mathcal{K}_{\mathrm{cand}}$\;
}
\EndIf
\end{algorithmic}
\begin{algorithmic}
\If{$\chi_{\mathrm{S}}=1$}
\state{
$\mathcal{K}_{\mathrm{leakage}}
\leftarrow
\textsc{LeakageGuidedCandidateReduction}
(\mathcal{K}_{\mathrm{cand}},\mathcal{L},\rho_{\min})$\;
}
\Else
\state
{
$\mathcal{K}_{\mathrm{leakage}}
\leftarrow
\mathcal{K}_{\mathrm{cand}}$\;
}
\EndIf
\end{algorithmic}

$\mathcal{K}_{\mathrm{eff}}
\leftarrow
\mathcal{K}_{\mathrm{cand}}
\cap
\mathcal{K}_{\mathrm{linear}}
\cap
\mathcal{K}_{\mathrm{differential}}
\cap
\mathcal{K}_{\mathrm{leakage}}$\;

$N_{\mathrm{eff}}
\leftarrow
|\mathcal{K}_{\mathrm{eff}}|$\;
\begin{algorithmic}
\If{$N_{\mathrm{eff}}=0$}
\state{
\Return Failure\;
}
\EndIf

\If{$\chi_{\mathrm{Q}}=0$}
\state{
Perform deterministic verification over
$\mathcal{K}_{\mathrm{eff}}$\;

Return valid candidate key\;
}
\EndIf
\end{algorithmic}
$q
\leftarrow
\left\lceil
\log_{2}(N_{\mathrm{eff}})
\right\rceil$\;

Construct reduced-space verification oracle
$U_{f}$ over
$\mathcal{K}_{\mathrm{eff}}$\;

$R
\leftarrow
\left\lfloor
\dfrac{\pi}{4}
\sqrt{N_{\mathrm{eff}}}
\right\rfloor$\;
\begin{algorithmic}
\For{$r=1$ To $R$}
\state{
Apply oracle phase inversion using $U_{f}$\;

Apply diffusion operator $G$\;
}
\EndFor
\end{algorithmic}
Measure computational register\;

Obtain candidate key $K^{*}$\;

Verify
$\mathcal{E}_{K^{*}}(P)=C$\;

\Return $K^{*}$\;

\end{algorithm}
\begin{lemma}
[Adaptive Candidate-Space Contraction under Capability-Driven Filtering]
\label{lem:adaptive_contraction}

Let $\mathcal{K}_{\mathrm{cand}}\subseteq\mathcal{K}_{k}$ denote the candidate-key space generated by \textsc{CandidateSpaceConstruction}$(\cdot)$ and let the activated cryptanalytic modules produce candidate subsets $\mathcal{K}_{\mathrm{linear}}$, $\mathcal{K}_{\mathrm{differential}}$, and $\mathcal{K}_{\mathrm{leakage}}$. Furthermore, let $\chi_{\mathrm{L}},\chi_{\mathrm{D}},\chi_{\mathrm{S}}\in\{0,1\}$ denote the capability-adaptive activation indicators defined in Eqs.~\eqref{eq:linear_activation}-\eqref{eq:leakage_activation}.

Define

\begin{equation}
\mathcal{K}_{\mathrm{L}}^{(\chi_{\mathrm{L}})}
=
\begin{cases}
\mathcal{K}_{\mathrm{linear}},
&
\chi_{\mathrm{L}}=1,
\\
\mathcal{K}_{\mathrm{cand}},
&
\chi_{\mathrm{L}}=0,
\end{cases}
\label{eq:lemma1_linear_set}
\end{equation}

\begin{equation}
\mathcal{K}_{\mathrm{D}}^{(\chi_{\mathrm{D}})}
=
\begin{cases}
\mathcal{K}_{\mathrm{differential}},
&
\chi_{\mathrm{D}}=1,
\\
\mathcal{K}_{\mathrm{cand}},
&
\chi_{\mathrm{D}}=0,
\end{cases}
\label{eq:lemma1_differential_set}
\end{equation}

and

\begin{equation}
\mathcal{K}_{\mathrm{S}}^{(\chi_{\mathrm{S}})}
=
\begin{cases}
\mathcal{K}_{\mathrm{leakage}},
&
\chi_{\mathrm{S}}=1,
\\
\mathcal{K}_{\mathrm{cand}},
&
\chi_{\mathrm{S}}=0.
\end{cases}
\label{eq:lemma1_leakage_set}
\end{equation}

The effective candidate-key space generated by Algorithm~\ref{alg:adaptive_hybrid_key_recovery} is

\begin{equation}
\mathcal{K}_{\mathrm{eff}}
=
\mathcal{K}_{\mathrm{cand}}
\cap
\mathcal{K}_{\mathrm{L}}^{(\chi_{\mathrm{L}})}
\cap
\mathcal{K}_{\mathrm{D}}^{(\chi_{\mathrm{D}})}
\cap
\mathcal{K}_{\mathrm{S}}^{(\chi_{\mathrm{S}})},
\label{eq:lemma1_effective_space}
\end{equation}

which satisfies

\begin{equation}
\mathcal{K}_{\mathrm{eff}}
\subseteq
\mathcal{K}_{\mathrm{cand}}
\subseteq
\mathcal{K}_{k},
\label{eq:lemma1_subset_relation}
\end{equation}

and

\begin{equation}
\left|\mathcal{K}_{\mathrm{eff}}\right|
\le
\min
\Bigl\{
\left|\mathcal{K}_{\mathrm{L}}^{(\chi_{\mathrm{L}})}\right|,
\left|\mathcal{K}_{\mathrm{D}}^{(\chi_{\mathrm{D}})}\right|,
\left|\mathcal{K}_{\mathrm{S}}^{(\chi_{\mathrm{S}})}\right|,
\left|\mathcal{K}_{\mathrm{cand}}\right|
\Bigr\}.
\label{eq:lemma1_cardinality_bound}
\end{equation}

Consequently, capability-adaptive filtering can only preserve or reduce the admissible candidate-key space and can never enlarge the cryptanalytic search domain.
\end{lemma}

\begin{proof}

By construction, each activated reduction procedure developed in Chapter~3 eliminates candidate keys that violate a corresponding statistical consistency criterion. Hence, for every activated module,

\[
\mathcal{K}_{i}
\subseteq
\mathcal{K}_{\mathrm{cand}},
\qquad
\mathcal{K}_{i}
\in
\left\{
\mathcal{K}_{\mathrm{linear}},
\mathcal{K}_{\mathrm{differential}},
\mathcal{K}_{\mathrm{leakage}}
\right\}.
\]

For inactive modules, Eqs.~\eqref{eq:lemma1_linear_set}-\eqref{eq:lemma1_leakage_set} assign the corresponding candidate subset to $\mathcal{K}_{\mathrm{cand}}$, thereby preserving the admissible search space without introducing additional hypotheses. Consequently,

\[
\mathcal{K}_{\mathrm{L}}^{(\chi_{\mathrm{L}})}
\subseteq
\mathcal{K}_{\mathrm{cand}},
\qquad
\mathcal{K}_{\mathrm{D}}^{(\chi_{\mathrm{D}})}
\subseteq
\mathcal{K}_{\mathrm{cand}},
\qquad
\mathcal{K}_{\mathrm{S}}^{(\chi_{\mathrm{S}})}
\subseteq
\mathcal{K}_{\mathrm{cand}}.
\]

Taking the intersection of these sets yields Eq.~\eqref{eq:lemma1_effective_space}. Since the intersection of subsets of a set remains a subset of that set,

\[
\mathcal{K}_{\mathrm{eff}}
\subseteq
\mathcal{K}_{\mathrm{cand}}.
\]

Furthermore, $\mathcal{K}_{\mathrm{cand}}\subseteq\mathcal{K}_{k}$ by definition of \textsc{CandidateSpaceConstruction}$(\cdot)$, establishing Eq.~\eqref{eq:lemma1_subset_relation}. For finite sets, the cardinality of an intersection cannot exceed the cardinality of any participating set. Therefore,

\[
\left|\mathcal{K}_{\mathrm{eff}}\right|
\le
\left|\mathcal{K}_{\mathrm{L}}^{(\chi_{\mathrm{L}})}\right|,
\qquad
\left|\mathcal{K}_{\mathrm{eff}}\right|
\le
\left|\mathcal{K}_{\mathrm{D}}^{(\chi_{\mathrm{D}})}\right|,
\qquad
\left|\mathcal{K}_{\mathrm{eff}}\right|
\le
\left|\mathcal{K}_{\mathrm{S}}^{(\chi_{\mathrm{S}})}\right|,
\]

and

\[
\left|\mathcal{K}_{\mathrm{eff}}\right|
\le
\left|\mathcal{K}_{\mathrm{cand}}\right|.
\]

Combining the above inequalities establishes Eq.~\eqref{eq:lemma1_cardinality_bound}. Therefore, adaptive capability-driven filtering is monotonic with respect to candidate-space cardinality and cannot enlarge the admissible cryptanalytic search space.
\end{proof}

Reduced-Space Quantum Verification Bound
\label{lem:reduced_space_quantum_bound}

Let

\begin{equation}
N_{\mathrm{eff}}
=
\left|
\mathcal{K}_{\mathrm{eff}}
\right|
\label{eq:lemma2_neff}
\end{equation}

denote the effective candidate-space cardinality generated by Algorithm~\ref{alg:adaptive_hybrid_key_recovery}, and let

\begin{equation}
U_f:
\mathcal{K}_{\mathrm{eff}}
\rightarrow
\{0,1\}
\label{eq:lemma2_oracle}
\end{equation}

be the reduced-space verification oracle constructed over the admissible candidate-key space. Assume that $M$ candidate states satisfy $f(K_i)=1$, where $M\geq1$. Then the number of amplitude-amplification iterations required to recover a valid key hypothesis with constant success probability satisfies

\begin{equation}
R_{\mathrm{eff}}
=
\theta
\!\left(
\sqrt{
\frac{N_{\mathrm{eff}}}{M}
}
\right),
\label{eq:lemma2_iteration_bound}
\end{equation}

and the corresponding quantum verification complexity becomes

\begin{equation}
T_Q
=
\theta
\!\left(
\sqrt{
\frac{N_{\mathrm{eff}}}{M}
}
\,
\left(
T_{U_f}
+
T_G
\right)
\right),
\label{eq:lemma2_quantum_complexity}
\end{equation}

where $T_{U_f}$ and $T_G$ denote the implementation costs of the verification oracle and diffusion operator, respectively. Furthermore, if

\begin{equation}
N_{\mathrm{eff}}
<
\left|
\mathcal{K}_k
\right|,
\label{eq:lemma2_reduction_condition}
\end{equation}

then

\begin{equation}
T_Q
<
\theta
\!\left(
\sqrt{
\frac{
\left|
\mathcal{K}_k
\right|
}{M}
}
\,
\left(
T_{U_f}
+
T_G
\right)
\right),
\label{eq:lemma2_reduction_bound}
\end{equation}
\begin{lemma}
establishing that adaptive candidate-space reduction strictly decreases the asymptotic quantum verification cost whenever at least one activated cryptanalytic module eliminates invalid key hypotheses.
\end{lemma}

\begin{proof}

Let

\begin{equation}
\mathcal{H}
=
\mathrm{span}
\left\{
|K_0\rangle,
|K_1\rangle,
\ldots,
|K_{N_{\mathrm{eff}}-1}\rangle
\right\}
\label{eq:lemma2_hilbert_space}
\end{equation}

denote the Hilbert space generated by the admissible candidate-key set. The reduced-space initialization stage prepares the uniform superposition

\begin{equation}
|s\rangle
=
\frac{1}{\sqrt{N_{\mathrm{eff}}}}
\sum_{K_i\in\mathcal{K}_{\mathrm{eff}}}
|K_i\rangle.
\label{eq:lemma2_uniform_superposition}
\end{equation}

Partition $\mathcal{H}$ into the marked and unmarked subspaces

\begin{equation}
\mathcal{H}_{M}
=
\mathrm{span}
\left\{
|K_i\rangle :
f(K_i)=1
\right\},
\qquad
\mathcal{H}_{U}
=
\mathrm{span}
\left\{
|K_i\rangle :
f(K_i)=0
\right\},
\label{eq:lemma2_partition}
\end{equation}

and let $|w\rangle$ and $|r\rangle$ denote normalized basis states spanning $\mathcal{H}_{M}$ and $\mathcal{H}_{U}$, respectively. The initial state can therefore be expressed as

\begin{equation}
|s\rangle
=
\sqrt{
\frac{M}{N_{\mathrm{eff}}}
}
\,|w\rangle
+
\sqrt{
1-
\frac{M}{N_{\mathrm{eff}}}
}
\,|r\rangle.
\label{eq:lemma2_state_decomposition}
\end{equation}

Define

\begin{equation}
\sin\theta
=
\sqrt{
\frac{M}{N_{\mathrm{eff}}}
}.
\label{eq:lemma2_theta}
\end{equation}

Each amplitude-amplification iteration performs a rotation of angle $2\theta$ within the invariant two-dimensional subspace generated by $|w\rangle$ and $|r\rangle$. Consequently, after $R$ iterations the state becomes

\begin{equation}
|\psi_R\rangle
=
\sin
\!\bigl(
(2R+1)\theta
\bigr)
|w\rangle
+
\cos
\!\bigl(
(2R+1)\theta
\bigr)
|r\rangle.
\label{eq:lemma2_rotated_state}
\end{equation}

To maximize the success probability, the marked-state amplitude must satisfy

\begin{equation}
(2R+1)\theta
\approx
\frac{\pi}{2}.
\label{eq:lemma2_success_condition}
\end{equation}

Substituting Eq.~\eqref{eq:lemma2_theta} into Eq.~\eqref{eq:lemma2_success_condition} yields

\begin{equation}
R
=
\theta
\!\left(
\sqrt{
\frac{N_{\mathrm{eff}}}{M}
}
\right),
\label{eq:lemma2_iteration_proof}
\end{equation}

thereby establishing Eq.~\eqref{eq:lemma2_iteration_bound}. Since each iteration invokes exactly one oracle operation and one diffusion operation, the total quantum verification cost is

\begin{equation}
T_Q
=
R
\left(
T_{U_f}
+
T_G
\right).
\label{eq:lemma2_cost_model}
\end{equation}

Substituting Eq.~\eqref{eq:lemma2_iteration_proof} into Eq.~\eqref{eq:lemma2_cost_model} establishes Eq.~\eqref{eq:lemma2_quantum_complexity}.

Furthermore, from Lemma~\ref{lem:adaptive_contraction},

\begin{equation}
N_{\mathrm{eff}}
\le
\left|
\mathcal{K}_k
\right|.
\label{eq:lemma2_contraction_reference}
\end{equation}

Whenever at least one activated reduction procedure removes invalid candidate keys,

\[
N_{\mathrm{eff}}
<
\left|
\mathcal{K}_k
\right|.
\]

Substituting this inequality into Eq.~\eqref{eq:lemma2_quantum_complexity} directly yields Eq.~\eqref{eq:lemma2_reduction_bound}. Therefore, adaptive candidate-space contraction translates directly into reduced oracle-evaluation requirements and reduced quantum verification complexity.

\end{proof}

\begin{theorem}[Adaptive Hybrid Cryptanalytic Complexity]
\label{thm:adaptive_hybrid_complexity}

Let Algorithm~\ref{alg:adaptive_hybrid_key_recovery} operate over the complete cryptographic key space $\mathcal{K}_{k}$ and let

\[
N_{\mathrm{eff}}
=
\left|
\mathcal{K}_{\mathrm{eff}}
\right|
\]

denote the effective candidate-space cardinality obtained after adaptive capability-driven filtering. Furthermore, let

\[
T_{\mathrm{CSC}},
\quad
T_{\mathrm{linear}},
\quad
T_{\mathrm{differential}},
\quad
T_{\mathrm{leakage}}
\]

represent the computational complexities of candidate-space construction, linear cryptanalytic reduction, differential cryptanalytic reduction, and leakage-guided reduction, respectively, as derived in Chapter~3. Then the total computational complexity of the adaptive hybrid cryptanalytic framework is

\begin{equation}
T_{\mathrm{AH}}
=
T_{\mathrm{CSC}}
+
\chi_{\mathrm{L}}
T_{\mathrm{linear}}
+
\chi_{\mathrm{D}}
T_{\mathrm{differential}}
+
\chi_{\mathrm{S}}
T_{\mathrm{leakage}}
+
T_Q,
\label{eq:theorem_total_complexity}
\end{equation}

where

\begin{equation}
T_Q
=
\theta
\!\left(
\sqrt{
\frac{N_{\mathrm{eff}}}{M}
}
\,
\left(
T_{U_f}
+
T_G
\right)
\right)
\label{eq:theorem_quantum_component}
\end{equation}

denotes the reduced-space quantum verification complexity and $M$ denotes the number of valid candidate-key hypotheses. Furthermore,

\begin{equation}
T_{\mathrm{AH}}
\le
T_{\mathrm{CSC}}
+
\chi_{\mathrm{L}}
T_{\mathrm{linear}}
+
\chi_{\mathrm{D}}
T_{\mathrm{differential}}
+
\chi_{\mathrm{S}}
T_{\mathrm{leakage}}
+
\theta
\!\left(
\sqrt{
\frac{
\left|
\mathcal{K}_{k}
\right|
}{M}
}
\,
\left(
T_{U_f}
+
T_G
\right)
\right),
\label{eq:theorem_upper_bound}
\end{equation}

with strict inequality whenever at least one activated cryptanalytic reduction stage eliminates invalid candidate-key hypotheses.

\end{theorem}

\begin{proof}

Algorithm~\ref{alg:adaptive_hybrid_key_recovery} consists of five computational stages:

\begin{enumerate}
\item candidate-space construction,
\item linear cryptanalytic reduction,
\item differential cryptanalytic reduction,
\item leakage-guided reduction,
\item reduced-space quantum verification.
\end{enumerate}

The first four stages contribute

\[
T_{\mathrm{CSC}},
\qquad
\chi_{\mathrm{L}}
T_{\mathrm{linear}},
\qquad
\chi_{\mathrm{D}}
T_{\mathrm{differential}},
\qquad
\chi_{\mathrm{S}}
T_{\mathrm{leakage}},
\]

respectively, where the activation indicators $\chi_{\mathrm{L}},\chi_{\mathrm{D}},\chi_{\mathrm{S}}\in\{0,1\}$ ensure that inactive modules contribute zero computational cost.

By Lemma~\ref{lem:adaptive_contraction},

\begin{equation}
\mathcal{K}_{\mathrm{eff}}
\subseteq
\mathcal{K}_{\mathrm{cand}}
\subseteq
\mathcal{K}_{k},
\label{eq:theorem_subset_relation}
\end{equation}

which implies

\begin{equation}
N_{\mathrm{eff}}
=
\left|
\mathcal{K}_{\mathrm{eff}}
\right|
\le
\left|
\mathcal{K}_{k}
\right|.
\label{eq:theorem_neff_bound}
\end{equation}

Furthermore, Lemma~\ref{lem:reduced_space_quantum_bound} establishes that the quantum verification stage requires

\[
T_Q
=
\theta
\!\left(
\sqrt{
\frac{N_{\mathrm{eff}}}{M}
}
\,
\left(
T_{U_f}
+
T_G
\right)
\right).
\]

Substituting the individual stage complexities yields Eq.~\eqref{eq:theorem_total_complexity}.

Using Eq.~\eqref{eq:theorem_neff_bound}, we obtain

\[
\sqrt{
\frac{N_{\mathrm{eff}}}{M}
}
\le
\sqrt{
\frac{
\left|
\mathcal{K}_{k}
\right|
}{M}
},
\]

which directly establishes Eq.~\eqref{eq:theorem_upper_bound}. Finally, if at least one activated reduction procedure eliminates an invalid candidate-key hypothesis, then

\[
N_{\mathrm{eff}}
<
\left|
\mathcal{K}_{k}
\right|.
\]

Since the square-root function is strictly monotonic over the positive real numbers,

\[
\sqrt{
\frac{N_{\mathrm{eff}}}{M}
}
<
\sqrt{
\frac{
\left|
\mathcal{K}_{k}
\right|
}{M}
},
\]

which implies

\[
T_Q
<
\theta
\!\left(
\sqrt{
\frac{
\left|
\mathcal{K}_{k}
\right|
}{M}
}
\,
\left(
T_{U_f}
+
T_G
\right)
\right).
\]

Therefore the adaptive hybrid framework achieves strictly lower quantum verification complexity whenever at least one activated cryptanalytic reduction stage contracts the admissible candidate-key space.
\end{proof}

\subsection{Computational complexity analysis}
\label{sec:complexity_analysis}

The computational complexity of Algorithm~\ref{alg:adaptive_hybrid_key_recovery} is obtained by aggregating the costs of candidate-space construction, capability evaluation, activated cryptanalytic reductions, candidate-space refinement, and reduced-space quantum verification. Let $T_{\mathrm{AH}}$ denote the total execution cost of the adaptive hybrid framework. The candidate-space construction stage invokes \textsc{CandidateSpaceConstruction}$(\cdot)$ over $\mathcal{K}_{k}$ and contributes

\begin{equation}
T_{\mathrm{CSC}}
=
\theta\!\left(
2^{k}T_D(n)
\right),
\label{eq:complexity_csc}
\end{equation}

where $k$ denotes the cryptographic key length and $T_D(n)$ denotes the computational cost of the underlying distinguisher evaluation over an $n$-bit observation. The capability-evaluation stage computes the activation indicators $\chi_{\mathrm{L}}$, $\chi_{\mathrm{D}}$, $\chi_{\mathrm{S}}$, and $\chi_{\mathrm{Q}}$. Since only a constant number of threshold comparisons are performed,

\begin{equation}
T_{\mathrm{ACT}}
=
\theta(1).
\label{eq:complexity_activation}
\end{equation}

When activated, the linear cryptanalytic reduction contributes

\begin{equation}
T_{\mathrm{linear}}
=
\theta\!\left(
N_KN
\right),
\label{eq:complexity_linear}
\end{equation}

where $N_K=|\mathcal{K}_{\mathrm{cand}}|$ and $N$ denotes the number of known plaintext-ciphertext observations. Similarly, the differential reduction stage contributes

\begin{equation}
T_{\mathrm{differential}}
=
\theta\!\left(
N_KNn
\right),
\label{eq:complexity_differential}
\end{equation}

where $n$ denotes the cipher block length and $N$ denotes the number of chosen plaintext pairs used during differential evaluation. The leakage-guided reduction stage contributes

\begin{equation}
T_{\mathrm{leakage}}
=
\theta\!\left(
N_KN
\right),
\label{eq:complexity_leakage}
\end{equation}

where $N$ denotes the number of leakage traces processed during correlation analysis. The adaptive refinement stage constructs $\mathcal{K}_{\mathrm{eff}}$ through set intersections. Let $N_{\mathrm{eff}}=|\mathcal{K}_{\mathrm{eff}}|$. The corresponding complexity satisfies

\begin{equation}
T_{\mathrm{INT}}
=
\theta\!\left(
N_{\mathrm{eff}}
\right).
\label{eq:complexity_intersection}
\end{equation}

From Lemma~\ref{lem:reduced_space_quantum_bound}, the reduced-space quantum verification stage requires

\begin{equation}
R_{\mathrm{eff}}
=
\theta\!\left(
\sqrt{\frac{N_{\mathrm{eff}}}{M}}
\right),
\label{eq:complexity_iterations}
\end{equation}

oracle-amplification iterations, where $M$ denotes the number of valid candidate-key hypotheses. Let $T_{U_f}$ and $T_G$ denote the implementation costs of the verification oracle and diffusion operator, respectively. The resulting quantum verification complexity becomes

\begin{equation}
T_Q
=
\theta\!\left(
\sqrt{\frac{N_{\mathrm{eff}}}{M}}
\,
\left(
T_{U_f}+T_G
\right)
\right).
\label{eq:complexity_quantum}
\end{equation}

Combining Eqs.~\eqref{eq:complexity_csc}-\eqref{eq:complexity_quantum}, the overall adaptive hybrid complexity is

\begin{equation}
\begin{aligned}
T_{\mathrm{AH}}
=
&
T_{\mathrm{CSC}}
+
T_{\mathrm{ACT}}
+
\chi_{\mathrm{L}}T_{\mathrm{linear}}
+
\chi_{\mathrm{D}}T_{\mathrm{differential}}
\\
&
+
\chi_{\mathrm{S}}T_{\mathrm{leakage}}
+
T_{\mathrm{INT}}
+
T_Q.
\end{aligned}
\label{eq:complexity_total}
\end{equation}

Substituting Eqs.~\eqref{eq:complexity_csc}, \eqref{eq:complexity_activation}, \eqref{eq:complexity_linear}, \eqref{eq:complexity_differential}, \eqref{eq:complexity_leakage}, \eqref{eq:complexity_intersection}, and \eqref{eq:complexity_quantum} into Eq.~\eqref{eq:complexity_total} yields

\begin{equation}
\begin{aligned}
T_{\mathrm{AH}}
=
&
\theta\!\left(
2^{k}T_D(n)
\right)
+
\theta(1)
+
\chi_{\mathrm{L}}
\theta\!\left(
N_KN
\right)
\\
&
+
\chi_{\mathrm{D}}
\theta\!\left(
N_KNn
\right)
+
\chi_{\mathrm{S}}
\theta\!\left(
N_KN
\right)
\\
&
+
\theta\!\left(
N_{\mathrm{eff}}
\right)
+
\theta\!\left(
\sqrt{\frac{N_{\mathrm{eff}}}{M}}
\,
\left(
T_{U_f}+T_G
\right)
\right).
\end{aligned}
\label{eq:complexity_expanded}
\end{equation}

Using Lemma~\ref{lem:adaptive_contraction}, $N_{\mathrm{eff}}\leq N_K\leq 2^{k}$. Therefore,

\begin{equation}
T_Q
\leq
\theta\!\left(
\sqrt{\frac{2^{k}}{M}}
\,
\left(
T_{U_f}+T_G
\right)
\right),
\label{eq:complexity_upper_quantum}
\end{equation}

with strict inequality whenever at least one activated reduction procedure removes invalid candidate-key hypotheses. Consequently, the adaptive framework replaces full-space quantum verification over $\mathcal{K}_{k}$ by reduced-space verification over $\mathcal{K}_{\mathrm{eff}}$, yielding a lower quantum verification cost whenever $N_{\mathrm{eff}}<2^{k}$. Neglecting lower-order terms, the asymptotic complexity of the proposed framework is therefore

\begin{equation}
T_{\mathrm{AH}}
=
\theta\!\left(
2^{k}T_D(n)
+
\chi_{\mathrm{L}}N_KN
+
\chi_{\mathrm{D}}N_KNn
+
\chi_{\mathrm{S}}N_KN
+
\sqrt{\frac{N_{\mathrm{eff}}}{M}}
\,
\left(
T_{U_f}+T_G
\right)
\right).
\label{eq:complexity_final}
\end{equation}

\subsection{Illustrative reduced-space quantum realization of the adaptive hybrid framework}

To illustrate the reduced-space quantum verification stage of Algorithm~\ref{alg:adaptive_hybrid_key_recovery}, consider a cryptanalytic scenario in which the activated reduction functions \textsc{LinearCandidateReduction}$(\cdot)$, \textsc{DifferentialCandidateReduction}$(\cdot)$, and \textsc{LeakageGuidedCandidateReduction}$(\cdot)$ have already been executed. Following the adaptive filtering process, the framework produces the effective candidate-key space

\begin{equation}
\mathcal{K}_{\mathrm{eff}}
=
\left\{
K_1,K_4,K_7,K_9,K_{13},K_{18},K_{22},K_{27}
\right\},
\label{eq:effective_example_space}
\end{equation}

where $K_{18}$ denotes the unique key hypothesis satisfying all statistical distinguishers and cryptographic verification conditions. Therefore,

\begin{equation}
N_{\mathrm{eff}}
=
|\mathcal{K}_{\mathrm{eff}}|
=
8.
\label{eq:effective_example_cardinality}
\end{equation}

According to Lemma~\ref{lem:adaptive_contraction}, $\mathcal{K}_{\mathrm{eff}}\subseteq\mathcal{K}_{k}$, while Lemma~\ref{lem:reduced_space_quantum_bound} guarantees that the subsequent quantum verification stage operates only on the reduced candidate space rather than the complete cryptographic key space. Since $N_{\mathrm{eff}}=8$, the quantum search register requires $q=\lceil\log_2(N_{\mathrm{eff}})\rceil=3$ qubits. The surviving candidate-key hypotheses are encoded according to

\begin{equation}
\begin{aligned}
|000\rangle &\leftrightarrow K_1,
&
|001\rangle &\leftrightarrow K_4,
&
|010\rangle &\leftrightarrow K_7,
&
|011\rangle &\leftrightarrow K_9,
\\
|100\rangle &\leftrightarrow K_{13},
&
|101\rangle &\leftrightarrow K_{18},
&
|110\rangle &\leftrightarrow K_{22},
&
|111\rangle &\leftrightarrow K_{27}.
\end{aligned}
\label{eq:encoding}
\end{equation}

Assume that $K_{18}$ represents the unique key hypothesis satisfying all cryptanalytic consistency constraints together with the final ciphertext verification condition. Consequently, $K^{*}=K_{18}\leftrightarrow |101\rangle$. Unlike conventional full-space quantum search formulations, each computational basis state corresponds to a cryptanalytically admissible key hypothesis generated by the adaptive reduction framework. Hence, amplitude amplification is performed exclusively within the statistically filtered candidate space. Let the quantum search register be defined as $R_K=\{q_2,q_1,q_0\}$. The initial quantum state is $|\psi_0\rangle=[1,0,0,0,0,0,0,0]^T$. Application of the state-preparation layer $H^{\otimes3}$ generates a uniform superposition over all surviving candidate-key hypotheses,

\[
|\psi_1\rangle
=
\frac{1}{\sqrt8}
[1,1,1,1,1,1,1,1]^T.
\]

The corresponding transformation matrix is

\begin{equation}
H^{\otimes3}
=
\frac{1}{\sqrt8}
\begin{bmatrix}
1&1&1&1&1&1&1&1\\
1&-1&1&-1&1&-1&1&-1\\
1&1&-1&-1&1&1&-1&-1\\
1&-1&-1&1&1&-1&-1&1\\
1&1&1&1&-1&-1&-1&-1\\
1&-1&1&-1&-1&1&-1&1\\
1&1&-1&-1&-1&-1&1&1\\
1&-1&-1&1&-1&1&1&-1
\end{bmatrix},
\label{eq:hadamard_tensor}
\end{equation}

yielding the state $|\psi_1\rangle$ above. The reduced-space verification oracle evaluates

\begin{equation}
f(K_i)
=
\begin{cases}
1,
&
K_i\in\mathcal{K}_{\mathrm{eff}}
\;\wedge\;
E_{K_i}(P)=C,
\\[1mm]
0,
&
\text{otherwise}.
\end{cases}
\label{eq:oracle_function}
\end{equation}

For the present example, only $K_{18}$ satisfies the verification criterion. The oracle therefore performs phase inversion exclusively on state $|101\rangle$, yielding

\begin{equation}
U_f
=
\mathrm{diag}(1,1,1,1,1,-1,1,1).
\label{eq:oracle_matrix}
\end{equation}

Application of the oracle produces

\[
|\psi_2\rangle
=
\frac{1}{\sqrt8}
[1,1,1,1,1,-1,1,1]^T.
\]

Defining $|s\rangle=\frac{1}{\sqrt8}\sum_{i=0}^{7}|i\rangle$, the diffusion operator becomes

\begin{equation}
G
=
2|s\rangle\langle s|-I
=
\frac{1}{4}
\begin{bmatrix}
-3&1&1&1&1&1&1&1\\
1&-3&1&1&1&1&1&1\\
1&1&-3&1&1&1&1&1\\
1&1&1&-3&1&1&1&1\\
1&1&1&1&-3&1&1&1\\
1&1&1&1&1&-3&1&1\\
1&1&1&1&1&1&-3&1\\
1&1&1&1&1&1&1&-3
\end{bmatrix}.
\label{eq:diffusion_matrix}
\end{equation}

Applying the diffusion transformation gives

\[
|\psi_3\rangle
=
G|\psi_2\rangle
=
\frac{1}{4\sqrt8}
[1,1,1,1,1,11,1,1]^T.
\]

The probability of recovering the target candidate key becomes

\begin{equation}
P(K_{18})
=
\left(
\frac{11}{4\sqrt8}
\right)^2
=
\frac{121}{128}
\approx
0.9453.
\label{eq:success_probability}
\end{equation}

The reduced-space amplification operator is $Q=GU_f$, and the complete quantum evolution may be summarized as

\begin{equation}
|\psi_0\rangle
\xrightarrow{H^{\otimes3}}
|\psi_1\rangle
\xrightarrow{U_f}
|\psi_2\rangle
\xrightarrow{G}
|\psi_3\rangle
\xrightarrow{\mathrm{Measure}}
|101\rangle
\leftrightarrow
K_{18}.
\label{eq:full_evolution}
\end{equation}

The corresponding Hamiltonian interpretation follows directly from the oracle and diffusion operators. The oracle Hamiltonian is

\begin{equation}
H_f
=
|101\rangle\langle101|,
\label{eq:oracle_hamiltonian}
\end{equation}

yielding

\begin{equation}
U_f
=
e^{-i\pi H_f},
\label{eq:oracle_hamiltonian_unitary}
\end{equation}

while the diffusion Hamiltonian is $H_s=|s\rangle\langle s|$, with

\begin{equation}
G
=
2H_s-I.
\label{eq:diffusion_hamiltonian_relation}
\end{equation}

Consequently, the adaptive framework first contracts the admissible candidate-key space through classical cryptanalytic reductions and subsequently performs quantum amplitude amplification only over the surviving candidate hypotheses. This realization directly illustrates the mechanism established in Lemma~\ref{lem:adaptive_contraction}, the reduced verification complexity derived in Lemma~\ref{lem:reduced_space_quantum_bound}, and the overall adaptive hybrid complexity characterized in Theorem~\ref{thm:adaptive_hybrid_complexity}. It therefore provides a physically realizable reduced-space quantum verification procedure corresponding to Algorithm~\ref{alg:adaptive_hybrid_key_recovery}.

\section{Quantum-Assisted Observation Sampling for Statistical Cryptanalysis}
The effectiveness of practical cryptanalytic procedures is often influenced by the quality and diversity of the observations available for statistical analysis. In realistic attack environments, complete evaluation of all available observations may be computationally inefficient or experimentally unnecessary. Consequently, observation selection is frequently performed during statistical bias estimation, differential-pair evaluation, leakage-trace analysis, model training, and hypothesis verification. Such selection procedures naturally arise within linear cryptanalysis, differential cryptanalysis, and side-channel analysis, where only a subset of the available observations may be processed during a particular evaluation stage. \cite{bib20}

\begin{figure*}[!t]
\centering
\includegraphics[width=\textwidth]{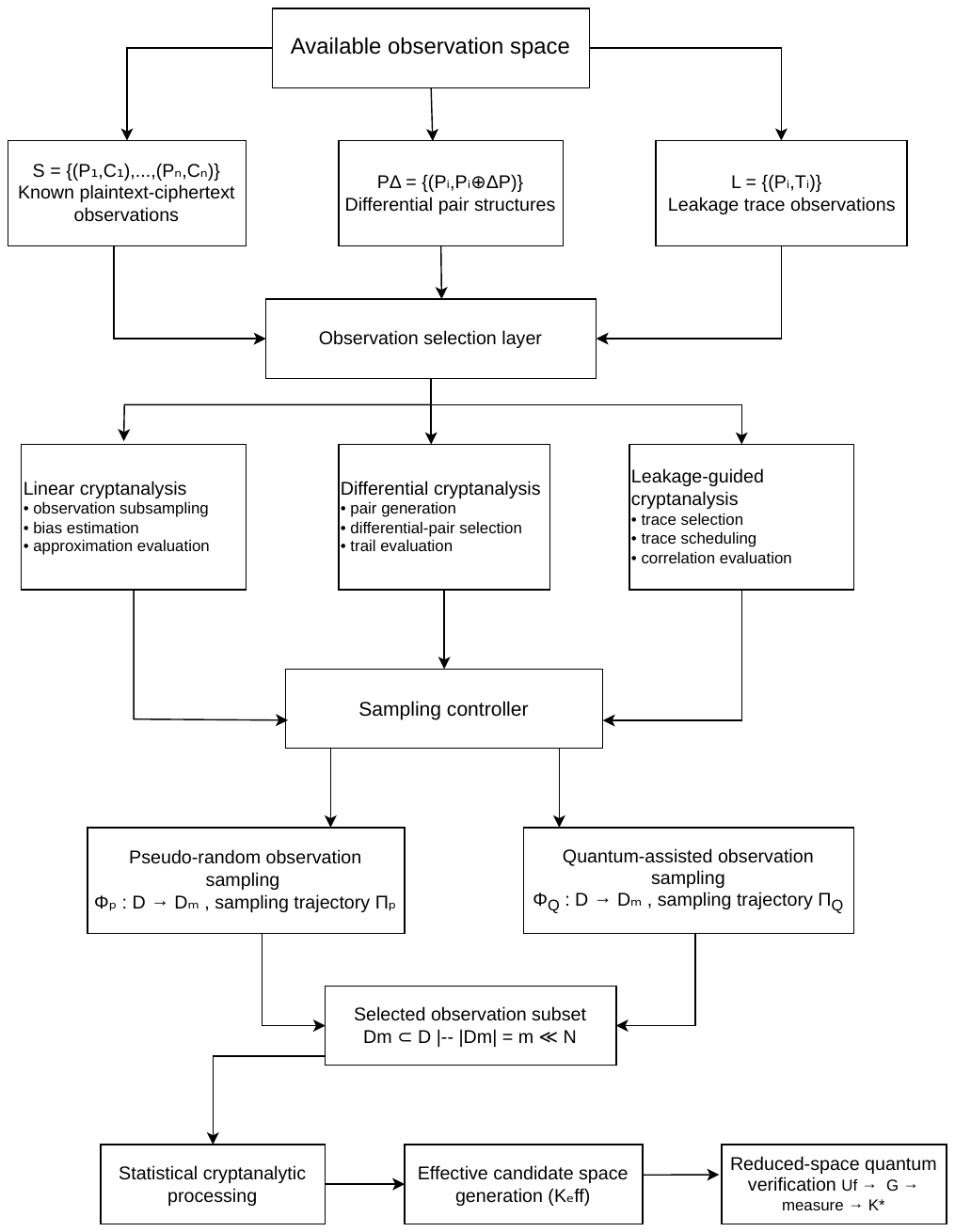}
\caption{
Quantum-assisted observation sampling within capability-adaptive cryptanalysis. The framework incorporates quantum-generated randomness exclusively during observation selection, differential-pair selection, and leakage-trace selection stages without altering the subsequent candidate-space reduction or reduced-space quantum verification procedures. \cite{bib20}
}
\label{fig:Quantum-assisted observation sampling within capability-adaptive cryptanalysis}
\end{figure*}

Let the complete observation space be represented by
\begin{equation}
\mathcal{D} =
\left\{
d_{1},
d_{2},
\ldots,
d_{N}
\right\},
\label{eq:observation_space}
\end{equation}
where each observation $d_i \in \mathcal{D}$ may correspond to a plaintext-ciphertext observation, a differential pair, or a leakage trace. A sampling operator $\Phi$ selects a reduced observation subset $\mathcal{D}_{m}\subseteq\mathcal{D}$ satisfying
\begin{equation}
\Phi:
\mathcal{D}
\rightarrow
\mathcal{D}{m},
\qquad
|\mathcal{D}{m}| = m,
\qquad
m \ll N.
\label{eq:sampling_operator}
\end{equation}
Within the proposed framework, observation selection may be required during linear approximation evaluation over the plaintext-ciphertext observation set
\begin{equation}
\mathcal{S} =
\left\{
(P_{1},C_{1}),
(P_{2},C_{2}),
\ldots,
(P_{N},C_{N})
\right\},
\label{eq:linear_observation_set}
\end{equation}
during differential-pair analysis over the structured pair set
\begin{equation}
\mathcal{P}_{\Delta} =
\left\{
(P_{i}, P_{i} \oplus \Delta P) \;|\; i = 1,2,\ldots,N
\right\},
\label{eq:differential_pair_set}
\end{equation}
and during leakage-guided evaluation over the trace collection
\begin{equation}
\mathcal{L}
\left\{
(P_{i},\tau_{i})
\right\},
\label{eq:leakage_trace_set}
\end{equation}
where $\tau_i$ denotes the measured leakage trace associated with observation $P_i$.
Conventional implementations typically employ pseudo-random selection mechanisms, producing a deterministic sampling trajectory
\begin{equation}
\Pi_{P}
\left\{
d_{i_{1}},
d_{i_{2}},
\ldots,
d_{i_{m}}
\right\}.
\label{eq:prng_sampling}
\end{equation}
Since the resulting trajectory is generated through deterministic algorithmic state evolution, the corresponding observation-selection process remains reproducible whenever the generator state becomes known.
To eliminate deterministic sampling dependence, the sampling operator may instead be driven by a quantum random number generator. \cite{bib19} Let
\begin{equation}
\mathcal{R}_{Q}
\left\{
q_{1},
q_{2},
\ldots,
q_{m}
\right\}
\label{eq:qrng_sequence}
\end{equation}
denote a sequence of quantum-generated random samples obtained from a physical entropy source. The corresponding observation-selection trajectory becomes
\begin{equation}
\Pi_{Q}
\left\{
d_{q_{1}},
d_{q_{2}},
\ldots,
d_{q_{m}}
\right\}.
\label{eq:qrng_sampling}
\end{equation}
The objective of quantum-assisted observation sampling is not to modify the asymptotic complexity of the underlying cryptanalytic procedure, nor to increase the cryptographic strength of the target cipher. Rather, its purpose is to provide a non-deterministic and unbiased mechanism for observation selection, thereby reducing potential sampling artefacts during statistical estimation, differential-pair evaluation, leakage-trace analysis, and candidate-space exploration.
Accordingly, the proposed framework permits the integration of quantum-generated randomness as an auxiliary observation-selection mechanism whenever statistical subsampling is required. \cite{bib23} Since the sampling stage operates independently of the candidate-space reduction procedures and the subsequent reduced-space quantum verification stage, the overall computational complexity of the framework remains unchanged. \cite{bib21} Nevertheless, the resulting observation subset $\mathcal{D}_{m}$ becomes independent of deterministic pseudo-random state evolution and may therefore provide a more neutral basis for cryptanalytic evaluation. The overall integration of quantum-assisted observation sampling within the proposed capability-adaptive cryptanalytic framework is illustrated in Figure~\ref{fig:qrng_sampling_framework}.

\section{Results}
To evaluate the proposed framework under controlled and reproducible conditions, an initial candidate-key hypothesis space containing (4{,}096) candidates was constructed. For each candidate $(K_i)$, three cryptanalytic evaluation metrics were recorded: linear approximation bias $(b_i)$, differential characteristic probability $(P_i)$, and leakage correlation coefficient $(\rho_i)$. These metrics were generated using bounded Gaussian distributions chosen to emulate the statistical variability commonly encountered in practical cryptanalytic evaluations. Specifically, $(b_i\sim\mathcal N(0.045,0.018^2))$, $(P_i\sim\mathcal N(0.32,0.12^2))$, and $(\rho_i\sim\mathcal N(0.52,0.16^2))$, with all samples constrained to their corresponding admissible ranges. Binary acceptance indicators $(\chi_L(K_i))$, $(\chi_D(K_i))$, and $(\chi_S(K_i))$ were subsequently derived using threshold values $(\tau_b=0.065)$, $(\tau_P=0.45), and (\tau_\rho=0.65)$, respectively. A candidate was retained within the effective candidate space only if it satisfied all active cryptanalytic criteria, i.e., $(\chi_F(K_i)=\chi_L(K_i)\land\chi_D(K_i)\land\chi_S(K_i))$. The resulting dataset serves as the basis for evaluating candidate-space contraction, adaptive filtering behavior, and reduced-space quantum verification within the proposed framework.

Application of the threshold-based cryptanalytic filters produced substantial candidate-space contraction. Starting from an initial hypothesis space of $4{,}096$ candidates, the linear, differential, and leakage-guided evaluations retained $557$, $528$, and $854$ candidates, respectively. The adaptive intersection stage subsequently reduced the admissible candidate space to only $13$ hypotheses satisfying all active cryptanalytic criteria. Consequently, the effective candidate space represents approximately $0.317\%$ of the original hypothesis space, corresponding to a reduction of approximately $99.683\%$. These results as shown in table~\ref{tab:candidate_space_contraction} indicate that the combined filtering strategy substantially contracts the candidate-key search space prior to quantum verification, thereby supporting the reduced-space verification model developed in Section~4.

\begin{table}[htbp]
\centering
\caption{Candidate-space contraction statistics}
\label{tab:candidate_space_contraction}

\small
\renewcommand{\arraystretch}{1.15}

\begin{tabular}{lccc}
\toprule

\textbf{Stage} &
\textbf{Candidate Count} &
\textbf{Retention Ratio} &
\textbf{Reduction Ratio}
\\

\midrule

Initial      & 4096 & 1.0000  & 0.0000  \\
Linear       & 557  & 0.1360  & 0.8640  \\
Differential & 528  & 0.1289  & 0.8711  \\
Leakage      & 854  & 0.2085  & 0.7915  \\
Effective    & 13   & 0.00317 & 0.99683 \\

\bottomrule
\end{tabular}

\end{table}

Figure~\ref{fig:candidate_space_contraction} illustrates the progressive contraction of the candidate-key hypothesis space produced by the adaptive cryptanalytic filtering framework.

\begin{figure}[htbp]
\centering
\includegraphics[width=0.90\linewidth]{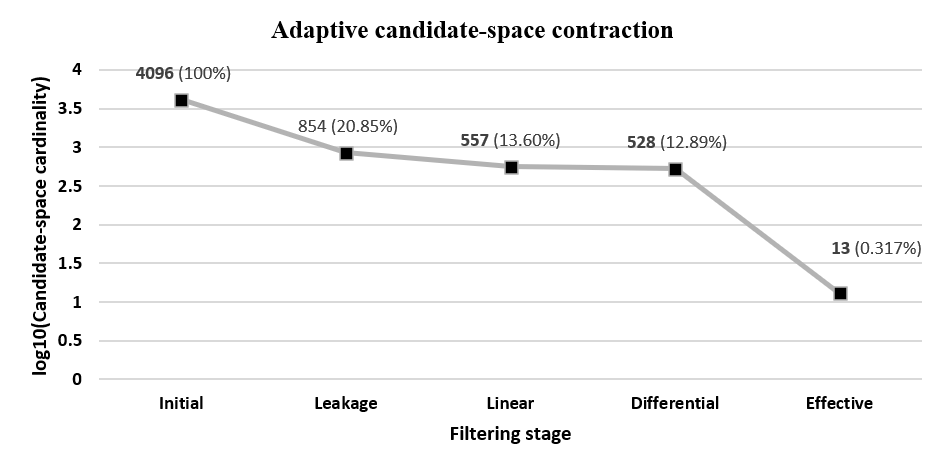}
\caption{
Progressive contraction of the candidate-key hypothesis space under the adaptive cryptanalytic filtering framework. Starting from an initial space of 4096 hypotheses, successive filtering stages substantially reduce the admissible candidate set, yielding an effective candidate space of only 13 hypotheses prior to quantum verification.
}
\label{fig:candidate_space_contraction}
\end{figure}

To quantify the impact of adaptive candidate-space contraction on quantum verification complexity, the Grover iteration requirement was evaluated as a function of the candidate-space cardinality. For a search space containing (N) candidate hypotheses, the required number of Grover iterations was computed using $R(N)=\left\lfloor \frac{\pi}{4}\sqrt{N}\right\rfloor.$

Unlike conventional quantum key-search formulations that operate over the complete search space, the proposed framework performs quantum verification only on the effective candidate space $(\mathcal{K}*{\mathrm{eff}})$ produced by the adaptive filtering stage. Consequently, the quantum verification complexity becomes dependent on $(N*{\mathrm{eff}})$ rather than the original candidate-space size (N). To illustrate this effect, Grover iteration statistics were evaluated for progressively reduced candidate spaces ranging from $(N=4096)$ to the experimentally obtained effective candidate space $(N_{\mathrm{eff}}=13)$. For each candidate-space size, the corresponding iteration count $(R(N))$, relative verification cost $(R(N)/R(4096))$, and reduction factor $(R(4096)/R(N))$ were computed. The resulting statistics provide a direct quantitative measure of the reduction in quantum verification effort achieved through adaptive candidate-space contraction.

Table~\ref{tab:quantum_verification_statistics}
summarizes the Grover iteration requirements and the corresponding verification-cost reduction achieved as the candidate-space size contracts.

\begin{table}[htbp]
\centering
\caption{Reduced-space quantum verification statistics}
\label{tab:quantum_verification_statistics}

\small
\renewcommand{\arraystretch}{1.15}

\begin{tabular}{cccc}
\toprule
\textbf{Candidate Space} &
\textbf{Grover} &
\textbf{Relative} &
\textbf{Reduction} \\
\textbf{Size ($N$)} &
\textbf{Iterations $R(N)$} &
\textbf{Cost} &
\textbf{Factor} \\
\midrule

4096 & 50 & 1.000 & 1.000 \\
2048 & 35 & 0.700 & 1.429 \\
1024 & 25 & 0.500 & 2.000 \\
512  & 17 & 0.340 & 2.941 \\
256  & 12 & 0.240 & 4.167 \\
128  & 8  & 0.160 & 6.250 \\
64   & 6  & 0.120 & 8.333 \\
32   & 4  & 0.080 & 12.500 \\
16   & 3  & 0.060 & 16.667 \\
13 ($N_{\mathrm{eff}}$) & 2 & 0.040 & 25.000 \\

\bottomrule
\end{tabular}
\end{table}

The results presented in table~\ref{tab:quantum_verification_statistics} demonstrate a substantial reduction in quantum verification complexity as the candidate space contracts. For the initial candidate space containing $4096$ hypotheses, the Grover search procedure requires $50$ iterations. \cite{bib15} Following adaptive cryptanalytic reduction, the effective candidate space contains only $13$ candidates, reducing the verification requirement to $2$ iterations. This corresponds to a relative verification cost of $0.04$ and a $25\times$ reduction in quantum verification effort. These observations are consistent with Lemma~2 and Theorem~1, which predict that candidate-space contraction directly translates into reduced quantum verification complexity. The results therefore provide empirical support for the reduced-space verification model underlying the proposed adaptive hybrid framework.

To evaluate the effectiveness of reduced-space quantum verification, the success probability of the target candidate hypothesis was analyzed under successive Grover amplification rounds. Using the reduced candidate space containing $(N_{\mathrm{eff}}=8)$ hypotheses, the target state probability was computed according to the standard amplitude-amplification relation $(P_r=\sin^2((2r+1)\theta))$, where $(\sin^2\theta=1/N_{\mathrm{eff}})$. The resulting statistics characterize the evolution of target-state amplitude during quantum verification and provide a quantitative illustration of the amplification process underlying the reduced-space search procedure.

\begin{table}[htbp]
\centering
\caption{Oracle success probability during reduced-space amplitude amplification}
\label{tab:oracle_success_probability}

\small
\renewcommand{\arraystretch}{1.15}

\begin{tabular}{cc}
\toprule

\textbf{Amplification Round ($r$)}
&
\textbf{Success Probability ($P_r$)}
\\

\midrule

0 & 0.125000 \\
1 & 0.781250 \\
2 & 0.945313 \\
3 & 0.330078 \\

\bottomrule
\end{tabular}

\end{table}

The results shown in table~\ref{tab:oracle_success_probability} demonstrate the expected behavior of quantum amplitude amplification. Beginning from an initial success probability of $0.125$, corresponding to a uniform superposition over the reduced candidate space, successive Grover iterations rapidly increase the probability of observing the correct candidate hypothesis. The maximum success probability occurs after two amplification rounds, reaching approximately $0.945$, after which additional iterations begin to overshoot the optimal amplitude configuration. This behavior is consistent with the reduced-space quantum verification model developed in Section~4 and illustrates how adaptive candidate-space contraction enables highly efficient quantum verification over a substantially smaller search space.

\section{Conclusion and Future Scope}\label{sec8}

This work presented a capability-adaptive cryptanalytic framework that integrates classical cryptanalytic filtering with reduced-space quantum verification. The proposed methodology combines linear cryptanalysis, differential cryptanalysis, and side-channel leakage analysis within a unified candidate-space reduction architecture, enabling cryptanalytic evidence from multiple attack modalities to be incorporated into a common decision framework. \cite{bib6} Rather than performing quantum verification over the complete key-search space, the framework constructs an effective candidate space through adaptive filtering and subsequently applies quantum amplitude amplification only to the surviving candidate hypotheses.

A formal mathematical characterization of the framework was developed through the candidate-space reduction model, adaptive verification algorithm, complexity analysis, and supporting theoretical results. The derived lemmas and theorem establish that candidate-space contraction directly reduces the verification effort required during the quantum search stage while preserving the admissible cryptanalytic solution space. The Hamiltonian formulation further provides a physically realizable interpretation of the reduced-space quantum verification procedure.

Experimental evaluation using statistically generated cryptanalytic observations demonstrated substantial candidate-space contraction. Starting from an initial hypothesis space containing $4096$ candidate keys, the adaptive filtering stage reduced the admissible candidate space to only $13$ hypotheses, corresponding to an overall reduction of approximately $99.683\%$. Consequently, the Grover verification requirement decreased from $50$ iterations to only $2$ iterations, yielding an approximately $25$-fold reduction in verification effort. Reduced-space amplitude amplification further demonstrated successful concentration of probability mass onto the target candidate hypothesis.

Future work may investigate integration of additional cryptanalytic modalities, including algebraic cryptanalysis, fault analysis, and protocol-level observations within the adaptive reduction framework. The proposed architecture may also be extended toward dynamic capability-aware cryptanalytic orchestration, larger-scale quantum verification workflows, and execution on contemporary quantum hardware platforms to further evaluate practical performance under realistic cryptanalytic conditions.
\bibliography{sn-bibliography}
\end{document}